\documentclass[runningheads]{llncs}

\usepackage[T1]{fontenc}
\usepackage[english]{babel}
\usepackage{indentfirst}

\usepackage{amsfonts}
\usepackage{graphicx}
\usepackage[colorlinks=true, allcolors=blue]{hyperref}
\usepackage[table]{xcolor}
\usepackage{comment}
\usepackage{mathtools}
\usepackage[ruled,linesnumbered,noend]{algorithm2e}
\usepackage{array}

\usepackage{thmtools}
\usepackage{thm-restate}

\newcommand{\CTM}{\approx_{CT}}
\newcommand{\nCTM}{\not\approx_{CT}}

\newcommand{\tauCTM}{\overset{\tau}{\approx}_{CT}}
\newcommand{\PDD}{\overrightarrow{PD}}
\newcommand{\PDG}{\overleftarrow{PD}}
\newcommand{\aD}{\overrightarrow{a_x}}
\newcommand{\bD}{\overrightarrow{b_x}}
\newcommand{\cD}{\overrightarrow{a_y}}
\newcommand{\dD}{\overrightarrow{b_y}}
\newcommand{\aG}{\overleftarrow{a_x}}
\newcommand{\bG}{\overleftarrow{b_x}}
\newcommand{\cG}{\overleftarrow{a_y}}
\newcommand{\dG}{\overleftarrow{b_y}}

\newcolumntype{P}[1]{>{\centering\arraybackslash}p{#1}}

\begin{document}
\title{Approximate Cartesian Tree Matching: an Approach Using Swaps}
%
%
\author{Bastien Auvray\inst{1} \and
Julien David\inst{1,2} \and
Richard Groult\inst{1,3} \and
Thierry Lecroq\inst{1,3}} 
\authorrunning{B. Auvray et al.}
%
\institute{CNRS NormaSTIC FR 3638, Caen, Le Havre, Rouen, France \and
Normandie University, UNICAEN, ENSICAEN, CNRS, GREYC,
Caen, France. \and
Univ Rouen Normandie, INSA Rouen Normandie, Université Le Havre Normandie, Normandie Univ, LITIS UR 4108, F-76000 Rouen, France}
\maketitle              
\begin{abstract}
Cartesian tree pattern matching consists of finding all the factors of a text that have
 the same Cartesian tree than a given pattern.
There already exist theoretical and practical solutions for the exact case.
In this paper, we propose the first algorithm for solving approximate Cartesian tree
 pattern matching.
We consider Cartesian tree pattern matching with one swap: given a pattern of length $m$
 and a text of length $n$
 we present two algorithms that find all the factors of the text that have the same Cartesian tree of
 the pattern after one transposition of two adjacent symbols.
The first algorithm uses a characterization of a linear representation of the Cartesian trees
 called parent-distance after one swap and runs in time $\Theta(mn)$ using $\Theta(m)$ space.
The second algorithm generates all the parent-distance tables of sequences
 that have the same Cartesian tree than the pattern after one swap.
It runs in time $\mathcal{O}((m^2 + n)\log{m})$
 and has $\mathcal{O}(m^2)$ space complexity.
\end{abstract}

\section{Introduction}

In general terms, the pattern matching problem consists of finding one or all
 the occurrences of a pattern $p$ of length $m$ in a text $t$ of length $n$.
When both the pattern and the text are strings the problem has been extensively
 studied and has received a huge number of solutions~\cite{FL2013}.
Searching time series or list of values for patterns representing specific
 fluctuations of the values requires a redefinition of the notion of pattern.
The question is to deal with the recognition of peaks, breakdowns, or more features.
For those specific needs one can use the notion of Cartesian tree. 

Cartesian trees have been introduced by Vuillemin in 1980~\cite{vuillemin80}.
They are mainly associated to strings of numbers and are structured as heaps
 from which original strings can be recovered by symmetrical traversals of the trees.
It has been shown that they are connected to Lyndon trees~\cite{CROCHEMORE20201},
 to Range Minimum Queries~\cite{demaine09}
 or to parallel suffix tree construction~\cite{SB14}.
Recently, Park \textit{et al.}~\cite{PALP19} introduced a new metric of generalized matching,
 called Cartesian tree matching.
It is the problem of finding every factor of a text $t$ which has the same
 Cartesian tree as that of a given pattern $p$.
Cartesian tree matching can be applied, for instance, to finding patterns in time series
 such as share prices in stock markets or gene sample time data.

Park \textit{et al.} introduced the parent-distance representation which is a
 linear form of the Cartesian tree and that has a one-to-one mapping with Cartesian trees.
They gave linear-time solutions for single and multiple pattern Cartesian tree matching,
 utilizing this parent-distance representation and existing classical string algorithms,
 i.e., Knuth-Morris-Pratt~\cite{KMP77} and Aho-Corasick~\cite{AC75} algorithms.
More efficient solutions for practical cases were given in~\cite{SGRFLP21}.
Recently, new results on Cartesian pattern matching appeared~\cite{PARK2020,KC2021,NFNI2021,FLPS22}.

All these previous works on Cartesian tree matching are concerned with finding exact occurrences
 of patterns consisting of contiguous symbols.
The only results known on non-contiguous symbols presents an algorithm for episode matching~\cite{OizumiKMIA22}
 (finding all minimal length factors of $t$ that contain $p$ as a subsequence) in Cartesian
 tree framework.

To the best of our knowledge, no result is known about approximate pattern matching in
 this Cartesian tree framework.
However in real life applications data are often noisy and it is thus important to
 find factors of the text that are similar, to some extent, to the pattern.
In this paper, we present the first results in this setting by considering
 approximate Cartesian tree pattern matching with one transposition
 (aka swap) of one symbol with the adjacent symbol.
Swap pattern matching has received a lot of attention in classical sequences
 since the first paper in 1997~\cite{AmirALLL97}
 (see~\cite{FaroP18} and references therein).
Swaps are common in real life data and it seems natural to consider them in the
 Cartesian pattern matching framework.
We are able to design two algorithms for solving the Cartesian tree pattern matching
 with at most one swap.
The first one runs in time $\Theta(mn)$ and uses a characterization of a linear representation
 of Cartesian trees while the second one runs in $\mathcal{O}((m^2 + n)\log{m} )$ and uses a graph to generate
 all the linear representations of Cartesian trees of sequences that match the pattern after one swap.

The remaining of the article is organized as follows: Sect.~\ref{sec:pre}
 presents the basic notions and notations used in the paper.
Given a sequence $x$, 
 in Sect.~\ref{sec:charac} we give a characterization of the parent-distance representations
 of Cartesian trees that correspond to sequences $x$ after one swap.
Sect.~\ref{sec:graph} presents the swap graph where vertices are Cartesian trees
 and there is an edge between two vertices if both Cartesian trees can be obtained from the other using one swap. 
In Sect.~\ref{sec:algo} we give our two algorithms for Cartesian tree pattern matching with swaps.
Sect.~\ref{sec:persp} contains our perspectives.

\section{Preliminaries \label{sec:pre}}

\subsection{Basic Notations}

We consider sequences of integers with a total order denoted by $<$. 
For a given sequence $x$, $\vert x\vert$ denotes the length of $x$, $x[i]$ is the $i$-th element of $x$ and $x[i\ldots j]$ represents the factor of $x$ starting at the $i$-th element and ending at the $j$-th element. 
For simplicity, we assume all elements in a given sequence to be distinct and numbered from 1 to $\vert x\vert$.

\subsection{Cartesian Tree Matching}

Given a sequence $x$ of length $n$, its Cartesian tree $C(x)$ is recursively defined as follows (see example~\figurename~\ref{fig:parent}):
\begin{itemize}
    \item if $x$ is empty, then $C(x)$ is the empty tree;
    \item if $x[1\ldots n]$ is not empty and $x[i]$ is the smallest value of $x$, $C(x)$ is the Cartesian tree with $i$ as its root, the Cartesian tree of $x[1\ldots i-1]$ as the left subtree and the Cartesian tree of $x[i+1\ldots n]$ as the right subtree.
\end{itemize}

\begin{figure}[h!]
\begin{center}
\includegraphics[scale=0.8]{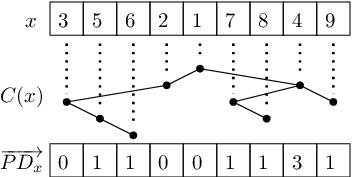}
\caption{
A sequence $z = (4, 5, 6, 2, 1, 7, 8, 3, 9) $, its Cartesian tree $C(z)$ and its corresponding
 parent-distance table $\PDD_z$.
}\label{fig:parent}    
\end{center}
\end{figure}
We will denote by $x \CTM y$ if sequences $x$ and $y$ share the same Cartesian tree. 
For example, $ z = (4, 5, 6, 2, 1, 7, 8, 3, 9) \CTM (3, 4, 8, 2, 1, 7, 9, 5, 6)$.

The Cartesian tree matching (CTM) problem consists in finding all factors of a text which share the same Cartesian tree as a pattern. 
Formally, Park \textit{et al.}~\cite{PALP19} define it as follows:

\begin{definition}[Cartesian tree matching] \label{ctm}
Given two sequences $p[1\ldots m]$ and $t[1\ldots n]$, find every $1 \le i \le n - m + 1$ such that $t[i\ldots i+m-1] \CTM p[1\ldots m]$.
\end{definition}

In order to solve CTM without building every possible Cartesian tree, an efficient representation of these trees was introduced by Park \textit{et al.}~\cite{PALP19}, the parent-distance representation (see example \figurename~\ref{fig:parent}):

\begin{definition}[Parent-distance representation] \label{def:pd}
Given a sequence $x[1\ldots n]$, the parent-distance representation of $x$ is an integer sequence $\PDD_x[1\ldots n]$, which is defined as follows:
\begin{align*}
\PDD_x[i] = 
\begin{cases*}
i - max_{1\leq j<i}\{j\ |\ x[j] < x[i]\}& \mbox{if such $j$ exists}\cr
0 & \mbox{otherwise}
\end{cases*}
\end{align*}
\end{definition}

Since the parent-distance representation has a one-to-one mapping with Cartesian trees, it can replace them without loss of information.

\subsection{Approximate Cartesian Tree Matching}

In order to define an approximate version of Cartesian tree matching, we use the following notion
of transposition on sequences: 

\begin{definition}[Swap]\label{def:swap}
Let $x$ and $y$ be two sequences of length $n$, and $i \in \{1,\ldots, n-1\}$, we denote $y = \tau(x,i) $ to describe a swap, that is:
$$y = \tau(x,i) \text{ if }\begin{cases}
x[j] = y[j], \forall j \notin \{i,i+1\}\\
x[i] = y[i+1] \\
x[i+1] = y[i]
\end{cases}
$$
\end{definition}
This kind of transposition is the one made by the Bubble Sort algorithm.
It is therefore a natural operation on permutations and sequences.
For the Cartesian tree point of view, see~\figurename\ref{fig:transposition}.
We use the notion of swap to define the approximate Cartesian tree matching.

\begin{definition}[$CT_\tau$ Matching] \label{ct-tau}
Let $x$ and $y$ be two sequences of length $n$, we have $x \tauCTM y$ if:
$$ \begin{cases} x \CTM y \text{, or}\\ 
\exists\ x',\ y', \exists\ i\in\{1,\ldots, n-1\}, x' \CTM x, y' \CTM y,  x' = \tau(y', i)\ and\ y' = \tau(x', i)\end{cases}$$
\end{definition}

\figurename~\ref{fig:transposition} shows an example of sequences that $CT_\tau$ 
 match.

\begin{figure}
    \centering
    \includegraphics{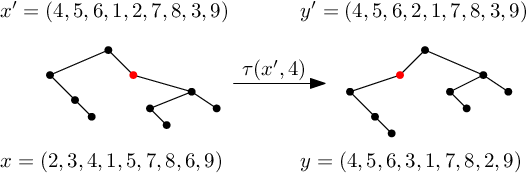}
    \caption{The sequence $x$ $CT_\tau$ matches $y$. 
    A swap at position $4$ moves the red node from the right subtree of the root to the left one. 
    In general, a swap at position $i$ consists either in moving
    the leftmost descendant of the right subtree to \textbf{a} rightmost position in 
    the left subtree (that is if $x[i] < x[i+1]$), or the opposite, in moving
    the rightmost descendant of the left subtree to
    \textbf{a} leftmost position of the right subtree of its parent.
    Note that we also have $x \tauCTM y'$, $x' \tauCTM y$ and of course $x' \tauCTM y'$.}
    \label{fig:transposition}
\end{figure}

Lastly, in order to fully characterize the approximate Cartesian tree matching, we introduce the notion of reverse parent-distance of a sequence that we compute as if read from right to left.
 
\begin{definition}[Reverse parent-distance] \label{rev-pd}
Given a sequence $x[1\ldots n]$, the reverse parent-distance representation of $x$ is an integer sequence $\PDG_x[1\ldots n]$, which is defined as follows:
$$\PDG_x[i] = \begin{cases}
min_{i < j \leq n}\{j\ |\ x[i] > x[j]\} - i & if\ such\ j\ exists\\
0 & otherwise
\end{cases}
$$
\end{definition}

\section{Characterization of the Parent-Distance Tables when a Swap Occurs \label{sec:charac}}

In this section, we describe how the parent-distances $\PDD_x$ and $\PDG_x$ of a sequence $x$ of length $n$ are modified into tables $\PDD_y$ and $\PDG_y$ when $y \tauCTM x$, with a swap occurring at position $i$.
\figurename~\ref{fig:zones} sums up the different parts of the parent-distance tables we are going to characterize. 
Those results will be used in Section~\ref{sec:algo} to obtain an algorithm that solves the $CT_\tau$ matching problem.

\begin{figure}[h!]
\begin{center}
\includegraphics[scale=0.3]{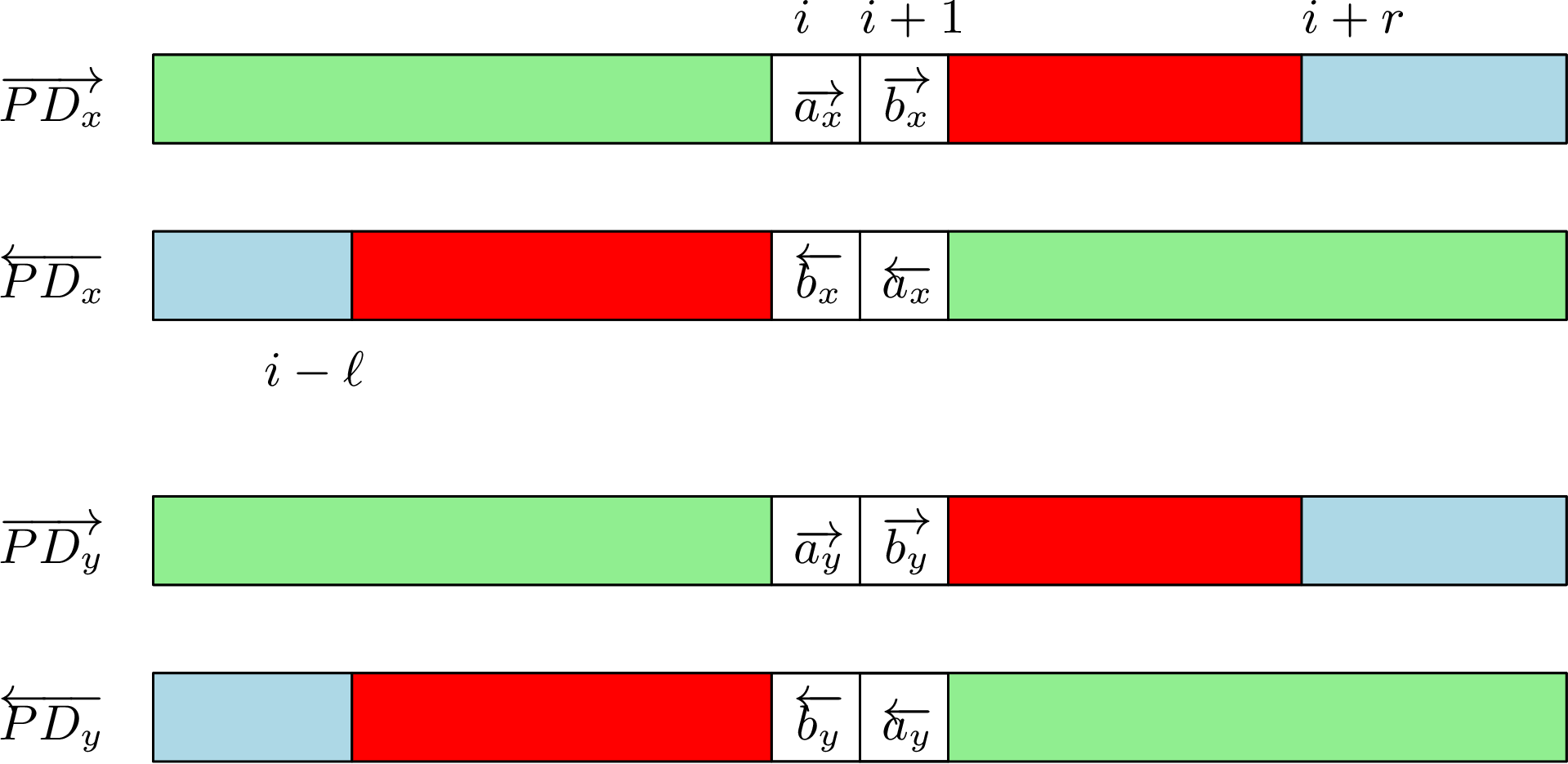}
\caption{This figure sums up the different Lemmas of this section. For instance, the green zones correspond to Def.~\ref{def:green} and Lemma~\ref{lm:green}. The values $\overrightarrow{a_x}$, $\overrightarrow{b_x}$, $\ldots$, are the $8$ values found in the parent-distance tables of $x$ and $y$ at position $i$ and $i+1$, that is $\overrightarrow{PD_x}[i] = \overrightarrow{a_x}$, $\overrightarrow{PD_x}[i+1] = \overrightarrow{b_x}$, $\ldots$ Values $i-\ell$ and $i+r$ respectively denote the last and first position of each blue zone.}\label{fig:zones}
\end{center}
\end{figure}

First, we describe how the parent-distances are modified at positions $i$ and $i+1$.
\begin{lemma}\label{lm:lesser}
Suppose that $x[i] < x[i+1]$, then the following properties hold:
\begin{enumerate}
    \item\label{lm:lesser:item1} $\dG = 1$
    \item\label{lm:lesser:item2} $\dD = \begin{cases}0 &\text{if }\aD = 0 \\ \aD + 1 & \text{otherwise}\end{cases}$
    \item\label{lm:lesser:item3} $\cG = \begin{cases}0 &\text{if } \bG = 0\\ \bG - 1 & \text{otherwise}\end{cases}$
    \item\label{lm:lesser:item4} $\cD \leq \begin{cases} i-1 & \text{if } \aD = 0\\ \aD & \text{otherwise}\end{cases}$
\end{enumerate}
\end{lemma}

\begin{proof}
Suppose $x[i] < x[i+1]$, we have $\bD = 1$ by definition of the parent-distance (Def.~\ref{def:pd}) and $\bG\neq1$ by definition of the reverse parent-distance
 (Def.~\ref{rev-pd}). Then, if a swap occurs at position $i$, $y[i] > y[i+1]$ and we have:
\begin{enumerate}
\item $\dG = 1$ by Def.~\ref{rev-pd}.

\item If $\aD = 0$,  $x[i]$ is the smallest element in $x[1\ldots i]$ by Def.~\ref{def:pd}. Which implies $y[i+1]$ is the smallest element in $y[1\ldots i+1]$ and thus $\dD = 0$ by Def.~\ref{def:pd}.

Otherwise, $x[i]$ (resp. $y[i+1]$) is not the smallest element in $x[1\ldots i]$ (resp. $y[1\ldots i+1]$). 
$y[i+1]$ has been pushed away from its parent in $y[1\ldots i-1]$ by one position compared to $x[i]$ and its parent in $x[1\ldots i-1]$. 
Thus, $\dD = \aD + 1$.

\item If $\bG = 0$, $x[i]$ is the smallest element in $x[i\ldots n]$ by Def.~\ref{rev-pd}. Which implies $y[i+1]$ is the smallest element in $y[i+1\ldots n]$, and thus $\cG = 0$ by Def.~\ref{rev-pd}.

Otherwise, $\bG > 1$ and $x[i]$ (resp. $y[i+1]$) is not the smallest element in $x[i\ldots n]$ (resp. $y[i+1\ldots n]$). 
$y[i+1]$ has been pushed closer to its parent in $y[i+2\ldots n]$ by one position when compared to $x[i]$ and its parent in $x[i+2\ldots n]$. 
Thus, $\cG = \bG - 1$.

\item If $\aD > 0$, that means there is an element smaller than $x[i]$ at position $i - \aD$ by Def.~\ref{def:pd}. After the swap, the parent-distance of $y[i]$ either refers to that same element at position $i - \aD$ or to a closer one that is smaller than $y[i]$ if such an element exists, and thus $\cD \leq \aD$.

Otherwise, the only information we have is $\cD \leq i - 1$ by Def.~\ref{def:pd}.
\end{enumerate}
\end{proof}

Note that in the item~\ref{lm:lesser:item4} of Lemma~\ref{lm:lesser} $\cD \leq \aD$, $\cD$ cannot take all values in $\{1, \ldots, \aD\}$,
since some won't produce a valid parent-distance table.

\begin{lemma}\label{lm:greater}
Suppose that $x[i] > x[i+1]$, then the following properties hold:
\begin{enumerate}
    \item\label{lm:greater:item1} $\dD = 1$
    \item\label{lm:greater:item2} $\dG = \begin{cases}0 &\text{if }\aG = 0 \\ \aG + 1 & \text{otherwise}\end{cases}$
    \item\label{lm:greater:item3} $\cD = \begin{cases}0 &\text{if } \bD = 0\\ \bD - 1 & \text{otherwise}\end{cases}$
    \item\label{lm:greater:item4} $\cG \leq \begin{cases} i-1 & \text{if } \aG = 0\\ \aG & \text{otherwise}\end{cases}$
\end{enumerate}
\end{lemma}
The proof is similar to the one of Lemma~\ref{lm:lesser}.
In the following, we define the green and blue zones of the parent-distances tables, which are equal, meaning that they are unaffected by the swap. Also, we define the red zones whose values differ by at most $1$. Proofs can be found in the appendix. We strongly invite the reader to use \figurename{}~\ref{fig:zones} to get a better grasp of the definitions.

\begin{definition}[The green zones]\label{def:green}
Given a sequence $x$ and a position $i$, 
the green zone of $\PDD_x$ is $\PDD_x[1\ldots i-1]$ and the green zone of $\PDG_x$ is $\PDG_x[i+2\ldots n]$.
\end{definition}

\begin{restatable}[The green zones]{lemma}{greenzones}
\label{lm:green}
The green zones of $\PDD_x$ and $\PDD_y$ (resp. $\PDG_x$ and $\PDG_y$) are equal.
\end{restatable}


\begin{definition}[The blue zones]\label{def:blue}
Given a sequence $x$ and a position $i$, the blue zone of 
$\PDD_x$ is $\PDD_x[i+r \ldots n]$ where:
$$
    r = 
    \begin{cases}
    \bG  &\text{if } x[i]<x[i+1] \text{ and } \bG > 1 \\  
    \aG+1  &\text{if } x[i] > x[i+1] \text{ and } \aG >0\\
    n-i+1 &\text{otherwise}
    \end{cases}
$$
The blue zone of  $\PDG_x$ is $\PDG_x[1 \ldots i-\ell]$ where:
$$
    \ell = 
    \begin{cases}
    \aD  &\text{if } x[i]<x[i+1] \text{ and } \aD>0 \\  
    \bD-1  &\text{if } x[i] > x[i+1] \text{ and } \bD>1\\
    i &\text{otherwise}
    \end{cases}
$$
\end{definition}

Note that in the last cases, the blue zones are empty.

\begin{restatable}[The blue zones]{lemma}{bluezones}
\label{lm:blue}
The blue zones of $\PDD_x$ and $\PDD_y$ (resp. $\PDG_x$ and $\PDG_y$) are equal.
\end{restatable}


\begin{definition}[The red zones]\label{def:red}
Given a sequence $x$ and a position $i$, if the blue zone of 
$\PDD_x$ is $\PDD_x[i+r \ldots n]$, then the right red zone is $\PDD_x[i+2 \ldots i+r-1]$.
Conversely, if the blue zone of $\PDG_x$ is $\PDG_x[1 \ldots i-\ell]$, then the left red zone is
$\PDG_x[i-\ell+1 \ldots i-1]$.
\end{definition}

\begin{restatable}[The red zones]{lemma}{redzones}
\label{lm:red}
We distinguish two symmetrical cases:
\begin{enumerate}
    \item $x[i] < x[i+1]$:
  \begin{enumerate}
    \item\label{lm:red:item1} 
  For each position $j$ in the red zone of $\PDD_x$, we have
  either $\PDD_y[j] = \PDD_x[j]$ or $\PDD_y[j] = \PDD_x[j]-1$.
    \item 
  For each position $j$ in the red zone of $\PDG_x$, we have
  either $\PDG_y[j] = \PDG_x[j]$ or $\PDG_y[j] = \PDG_x[j]+1$.
  \end{enumerate}
    \item $x[i] > x[i+1]$:
  \begin{enumerate}
    \item 
  For each position $j$ in the red zone of $\PDD_x$, we have
  either $\PDD_y[j] = \PDD_x[j]$ or $\PDD_y[j] = \PDD_x[j]+1$.
    \item 
  For each position $j$ in the red zone of $\PDG_x$, we have
  either $\PDG_y[j] = \PDG_x[j]$ or $\PDG_y[j] = \PDG_x[j]-1$.
  \end{enumerate}
\end{enumerate}
\end{restatable}

We now show that swaps at different positions produce different Cartesian trees.

\begin{lemma}\label{lm:uniondisjointe}
Let $x$ be a sequence of length $n$ and $i,j \in \{1,\ldots, n-1\}$, with $i\neq j$. 
Then $\tau(x,i) \nCTM \tau(x,j)$.
\end{lemma}
\begin{proof}
Suppose without loss of generality that $j > i $.
If $j>i+1$, then according to Lemma~\ref{lm:green}, we have:
$$\forall k< j, \PDD_x[k] = \PDD_{\tau(x,j)}[k] = \PDD_{\tau(x,i)}[k] $$
And according to Lemma~\ref{lm:lesser}, we have that $\PDD_x[i+1] \neq \PDD_{\tau(x,i)}[i+1]$,
which leads to a contradiction.

Then suppose that $j=i+1$, then it is sufficient to consider what happens on a sequence of length $3$:
having local differences on the parent-distance tables implies having different parent-distances and therefore do not $CT$ match. One can easily check that the lemma is true for each sequence of length $3$.

\end{proof}

\section{Swap Graph of Cartesian Trees \label{sec:graph}}

Let $\mathcal{C}_n$ be the set of Cartesian trees with $n$ nodes, which is equal to the set of binary trees with $n$ nodes. Also $\mathcal{C} = \bigcup_{n\ge 0} \mathcal{C}_n$.
Let $\mathcal{G}_n = (\mathcal{C}_n, E_n)$ be the Swap Graph of Cartesian trees,
where $\mathcal{C}_n$ is its set of vertices and $E_n$ the set of edges.
Let $x$ and $y$ be two sequences, we have $\{C(x),C(y)\} \in E_n$ if $x \tauCTM y$.
\figurename~\ref{fig:graph} shows the Swap Graph with $n$ smaller than $4$. 

\begin{figure}[h!]
\begin{center}
\begin{minipage}{0.49\textwidth}
\begin{center}
    \includegraphics{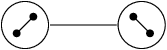}\\
    \vspace{5em}

\includegraphics{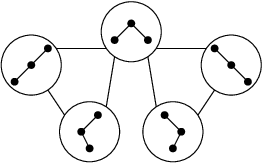}
\end{center}

\end{minipage}
\begin{minipage}{0.5\textwidth}
\includegraphics[scale=0.9]{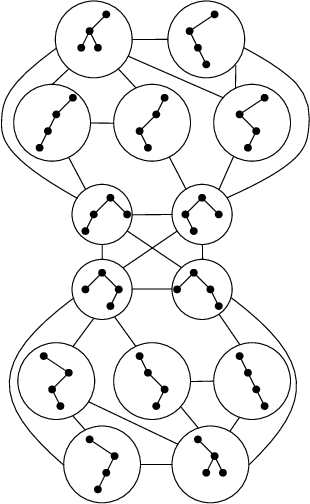}    
\end{minipage}
\end{center}
\caption{Swap Graph of Cartesian trees of size $2, 3$ and $4$. \label{fig:graph}}
\end{figure}

In the following, we study the set of neighbors a vertex can have in the Swap Graph.
Let $T \in \mathcal{C}_n$ be a Cartesian tree of size $n$ and $ng(T)$ be its set of neighbors in the Swap Graph. 
Also, for $i < n$, we note $ng(T,i)$ the set of trees obtained by doing a swap at position $i$ on the associated sequences, that is
$$ng(T,i) = \{ C(y) \in \mathcal{C}_n \mid \exists\ x\text{ such that  }T=C(x) \text{ and } y=\tau(x,i)  \}$$
Also, we have
$$ng(T) = \bigcup_{i=1}^{n-1} ng(T,i)$$
where all unions are disjoint according to Lemma~\ref{lm:uniondisjointe}.

\begin{lemma}\label{lm:decomp}
    Let $T \in \mathcal{C}_{n}$ be a Cartesian tree of size $n$, with a left-subtree $A$ of size $k-1$
    and $B$ a right subtree of size $n-k$. 
    We have
    $$|ng(T)| = |ng(A)| + |ng(B)| + |ng(T,k-1)| + |ng(T,k)|$$
\end{lemma}
\begin{proof}
    The result follows from Lemma~\ref{lm:uniondisjointe} and the definition of $ng(T)$.
    Indeed, we have $|ng(A)| = |\bigcup_{i=1}^{k-2} ng(T,i)|$ and $|ng(B)| = |\bigcup_{i=k+1}^{n-1} ng(T,i)|$.
\end{proof}

Let $lbl: \mathcal{C} \mapsto \mathbb{N}$ be the function that computes the length of the left-branch of a tree such that 
$$\forall\ T\in \mathcal{C}, lbl(T)= 
\begin{cases} 
0, \text{ if } T \text{ is empty}\\ 
1+lbl(A), \text{ where } A \text{ is the left-subtree of }T.
\end{cases}$$
The equivalent $lbr$ function can be defined to compute the length of the right-branch.

\begin{lemma}\label{lm:decompT}
    Let $T \in \mathcal{C}_{n}$ be a Cartesian tree with a root $i$, a left-subtree $A$ and a right-subtree $B$.
    $$|ng(T,i-1)| = lbl(B)+1 \text{ and } |ng(T,i)| = lbr(A)+1$$
\end{lemma}
\begin{proof}
    We only prove that $|ng(T,i-1)| = lbl(B)+1$ since the rest of the proof uses the same arguments.
    Let $x$ be a sequence such that $C(x) = T$.
    As stated in the definition section (see \figurename~\ref{fig:transposition}), the swap $\tau(x,i-1)$ moves the rightmost node of $A$ into a leftmost position in $B$. Let $j_1, \ldots, j_{lbl(B)}$ be the positions
    in the sequence $x$ that corresponds to the nodes of the left branch of $B$. For each $\ell < lbl(B)$, there always
    exists a sequence $y = \tau(x,i)$ such that $y[i] < y[j_1] < \cdots < y[j_\ell] < y[i-1] < y[j_{\ell+1}] < \cdots < y[j_{lbl(B)}]$. Therefore, there exist exactly $lbl(B)+1$ possible output trees when applying such a swap. 
\end{proof}

\begin{lemma}\label{lm:complete}
Let $T_h$ be the complete binary tree of height $h>0$, we have 
$$|ng(T_h)| = 6(2^h-1)-2h$$
\end{lemma}
\begin{proof}
Let us first remark that a complete binary tree of height $0$ is simply a leaf, thus $ng(T_0) = \emptyset$.
In the following, we consider $h>0$.
Using the fact we are dealing with complete trees and Lemma~\ref{lm:decomp}, we have:
    $$|ng(T_h)| = 2|ng(T_{h-1})| + |ng(T_h,2^h-1)| + |ng(T_h, 2^h )|$$
Also, using Lemma~\ref{lm:decompT} we have
    $$|ng(T_h)| = 2|ng(T_{h-1})| + 2(h+1) $$
Using the telescoping technique we obtained that
$$|ng(T_h)| = \sum_{i=0}^{h-1} (i+2)2^{h-i}$$
Which simplifies into
$$|ng(T_h)| = 6(2^h-1)-2h$$
\end{proof}
This lemma can be reformulated: let $n = 2^{h+1}-1$ be the number of nodes in the complete tree, we have
$$|ng(T_h)| = \lceil 3(n-1)-2(\log_2(n+1)-1) \rceil$$

\begin{lemma}\label{lm:neighbor}
For every Cartesian tree $T$ of size $n$, we have 
$$n-1 \le |ng(T)| \le \lceil 3(n-1)-2(\log_2(n+1)-1) \rceil$$
\end{lemma}
\begin{proof}
The first part of the inequality is given by Lemma~\ref{lm:uniondisjointe}: each tree has at least
$n-1$ neighbors, since each transposition at a position in $\{1,\ldots, n-1\}$ produces different Cartesian trees. 
The upper bound is obtained by the following reasoning: each Cartesian tree can be obtained
starting from a complete binary tree by iteratively removing some leaves. 
Let $T$ be a Cartesian tree for which the upper bound holds and $T'$ the tree obtained after
removing a leaf from $T$.
\begin{itemize}
    \item 
If the leaf $i$ had a right sibling (resp. left sibling), then $|ng(T,i)| = 2$ (resp. $|ng(T,i-1)| = 2$) and $|ng(T,i+1)| \ge 2$  (resp. $|ng(T,i)| \ge 2$). In $T'$, the neighbors from $ng(T,i)$ (resp. $ng(T,i-1)$) are lost and $|ng(T',i)|=|ng(T,i+1)|-1$ (resp. $|ng(T',i-1)|=|ng(T,i)|-1$). Therefore at $|ng(T')| \le |ng(T)|-3$.
\item 
If the leaf $i$ had no sibling, then there are two possibilities:
\begin{itemize}
    \item either the leaf were at an extremity of the tree (the end of the left/right branch). In this case, removing the leaf will only remove $2$ neighbors. This can only be made twice for each level of the complete tree (only one time for the first level), that is $log_2(n+1)-1$ times.
    \item or the leaf is in ``the middle'' of the tree, in this case the tree is at least of height $2$
    and we can remove at least $3$ neighbors.
\end{itemize}
\end{itemize}
\end{proof}

We use the previous lemma to obtain a lower bound on the diameter of the Swap Graph.

\begin{lemma}
    The diameter of the Swap Graph $\mathcal{G}_n$ is $\Omega(\frac{n}{\ln{n}})$.
\end{lemma}
\begin{proof}
    The number of vertices in the graph is equal to the number of binary trees enumerated by the Catalan numbers, that is $\frac{\binom{2n}{n}}{n+1}$. Since the maximal degree of a vertex is less than $3n$ according to Lemma~\ref{lm:neighbor}, the diameter is lower bounded by the value $k$ such
    that:
    $$(3n)^k = \frac{\binom{2n}{n}}{n+1}$$
    $$\implies k = \frac{\ln\left(\frac{\binom{2n}{n}}{n + 1} \right)}{\ln(3) + \ln(n)}$$
    $$\implies k = \frac{2n\ln(2n) - 2n\ln(n) - \ln(n + 1)}{\ln(3) + \ln(n)}$$
    By decomposing $2n\ln(2n)$ into $2n\ln{2} + 2n\ln{n}$ we obtain
    $$\implies k = \frac{2n\ln(2)}{\ln(3) + \ln(n)} - \frac{\ln(n + 1)}{\ln(3) + \ln(n)}$$
    $$\implies k = \Omega\left(\frac{n}{\ln{n}} \right)$$
\end{proof}

\section{Algorithms \label{sec:algo}}

In this section, we present two algorithms to compute the number of occurrences of $\tauCTM$ matching
of a given pattern in a given text.
The first algorithm uses two parent-distance tables and the set of Lemmas presented in Section~\ref{sec:charac}. 
The second algorithm uses Lemma~\ref{lm:neighbor} and the fact that each sequence has a bounded number of Cartesian trees that can be obtained from a swap to build an automaton.

\subsection{An Algorithm Using the Parent-Distance Tables}

Algorithm~\ref{algo:parent}, below, is based on Lemmas~\ref{lm:lesser} to~\ref{lm:uniondisjointe}.
The function candidateSwapPosition at 
 line~\ref{algo:parent:linei} computes the exact location of the swap, using the green zones\footnote{detailed in the Appendix}.
There can be up to two candidates for the position $i$ of the swap but only one can be valid according to Lemma~\ref{lm:uniondisjointe}. 

\begin{algorithm}\label{algo:parent}
  \caption{$DoubleParentDistanceMethod(p, t)$\label{algo:pds}}
  \SetAlgoLined
  \SetKwInOut{KwIn}{Input}
  \SetKwInOut{KwOut}{Output}
  \KwIn{Two sequences $p$ and $t$}
  \KwOut{The number of positions $j$ such that $t[j\ldots j+p-1] \tauCTM p$}
  $occ \leftarrow 0$\;
  $(\PDD_p, \PDG_p) \leftarrow$ Compute the parent-distance tables of $p$\;
  \For{$j \in \{1, \ldots, |t|-|p|+1\}$}{
    $(\PDD_x, \PDG_x) \leftarrow$ Compute the parent-distance tables of $x=t[j\ldots j+p-1]$\;  
    \If{$\PDD_p = \PDD_x$}{
        $occ \leftarrow occ +1$\;
    }
    \Else{
    \ForEach{$i \in candidateSwapPosition(\PDD_p, \PDG_p, \PDD_x, \PDG_x)$\label{algo:parent:linei}}{
            \If{Lemmas~\ref{lm:lesser},~\ref{lm:greater},~\ref{lm:blue} and~\ref{lm:red} hold for $p, x$ and $i$}{
                $occ \leftarrow occ +1$\;
            }
        }
    }
  }
  \Return{$occ$}\;
\end{algorithm}

\begin{theorem}
Given two sequences $p$ and $t$ of length $m$ and $n$, Algorithm~\ref{algo:pds}
 has a $\Theta(mn)$ worst-case time complexity
 and a $\Theta(m)$ space complexity.
\end{theorem}
\begin{proof}
The space complexity is obtained from the length of the parent-distance tables.
Regarding the time complexity, in the worst-case scenario, one has to check the zones of each parent-distance
table, meaning that the number of comparisons is bounded by $2m$ for each position
in the sequence $t$.
This worst case can be reached by taking a sequence $t$ whose associated Cartesian tree is a comb and a sequence $p$ obtained by doing a swap on a sequence associated to a comb. 
\end{proof}

\subsection{An Aho-Corasick Based Algorithm}

The idea of the second method is to take advantage of the upper bound on
the size of the neighborhood of a given Cartesian tree in the Swap Graph.
Given a sequence $p$, we compute the set of its neighbors $ng(C(p))$,
then we compute the set of all parent-distance tables and build the 
automaton that recognizes this set of tables using the Aho-Corasick method
 for multiple Cartesian tree matching~\cite{PALP19}.
Then, for each position in a sequence $t$, it is sufficient to read
the parent-distance table of each factor $t[j\ldots j+|p|-1]$ into the automaton
and check if it ends on a final state.
The following theorem can be obtained from Section $4.2$ in~\cite{PALP19}.
\begin{theorem}
Given two sequences $p$ and $t$ of length $m$ and $n$, the Aho-Corasick based
 algorithm has an $\mathcal{O}((m^2 + n)\log{m} )$ worst-case time complexity
 and an $\mathcal{O}(m^2)$ space complexity.
\end{theorem}

\section{Perspectives \label{sec:persp}}

From the pattern matching point of view, the first step would be to generalize
our result to sequences with a partial order instead of a total one.
Then, it could be interesting to obtain a general method, where the number
of swaps is given as a parameter. Though, we fear that if too many
swaps are applied, the result loses its interest, even though the complexity 
might grow rapidly.

The analysis of both algorithms should be improved.
An amortized analysis could be done on the first algorithm and on the computation
of the set of parent-distance tables for the second algorithm.
In the second algorithm the upper bounds on the time and space complexity 
could be improved by studying the size of the minimal automaton.

\bibliographystyle{splncs04}
\bibliography{biblio}

\newpage
\appendix
\section{Appendix \label{sec:app}}
In this section we show how to compute the position of the swap thanks to the green zones as mentioned in line~\ref{algo:parent:linei} of algorithm~\ref{algo:pds} and we provide proofs and additional context for Lemmas~\ref{lm:green},~\ref{lm:blue} and~\ref{lm:red}.

\subsection{Computing the Position of the Swap with the Green Zones}

We compute the candidate positions $i$ for the possible swap.
If the parent-distances are equal, then both sequences trivially $CT$ match.
Otherwise, the idea is to rely on a ``pincer movement'' using the green zones. According to Lemma~\ref{lm:green}, the green zones of
$\PDD_x$ and $\PDD_y$ (resp. $\PDG_x$ and $\PDG_y$) are equal. From Lemmas~\ref{lm:lesser} and~\ref{lm:greater}, we also deduce that $\bD \neq \dD$ and $\bG \neq \dG$. 
Unfortunately, no hard condition can be set on the $a$'s, and as such, we might run into four different cases,
depending on whether $\aD$ equals $\cD$ and $\aG$ equals $\cG$. Fig.~\ref{fig:indistinguability} presents all four possible cases.

\begin{figure}[h!]
    \centering
    \includegraphics[scale=0.55]{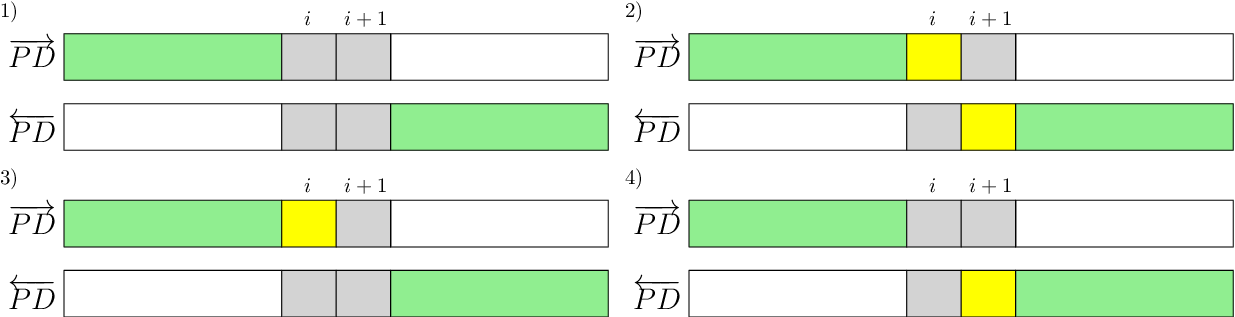}
    \caption{The parent-distance tables of sequences $x$ and $y$ are merged into one. If a zone is colored either in green or yellow, then the tables match, in grey if they do not and white when it is unknown.}
    \label{fig:indistinguability}
\end{figure}

In case 1, there is a gap of length 2 between the green/yellow zones, and we can immediately deduce that the positions of this gap are the only eligible positions for the swapped elements. In case 2, the gap is reduced to 0, but we can once again pinpoint the position $i$ with total accuracy. Lastly, in cases 3 and 4, there is a gap of 1,
and as we currently have no knowledge of the position of the swap, both cases end up indistinguishable,
meaning that we must test the positions on the left of the gap and of the gap itself. From Lemma~\ref{lm:uniondisjointe}, we have that only one such position 
might be a swap, as two swaps on the same sequence would produce different Cartesian trees and thus different parent-distance tables.
Any gap larger than 2 immediately disqualifies the sequences from $CT_\tau$ matching.

\subsection{Further Illustrations and Proofs of Lemmas~\ref{lm:green},~\ref{lm:blue} and~\ref{lm:red}}

The idea behind the blue and red zones is to use both parent-distance tables to isolate another part of the 
parent-distance tables that is not affected by the swap when the latter happens in a subtree of the Cartesian tree.
We observe that the smallest element involved in the swap is always the root of the subtree where the swap happens, with
the other element being either the leftmost element of the root's right subtree, or the rightmost element of the root's 
left subtree. We then consult both parent-distances of the smallest element to locate the last and first elements of the
blue zones, if those blue zones exist. 

\begin{figure}[h!]
\begin{minipage}[t]{0.49\textwidth}
\vspace{0pt}
\includegraphics[scale=0.35]{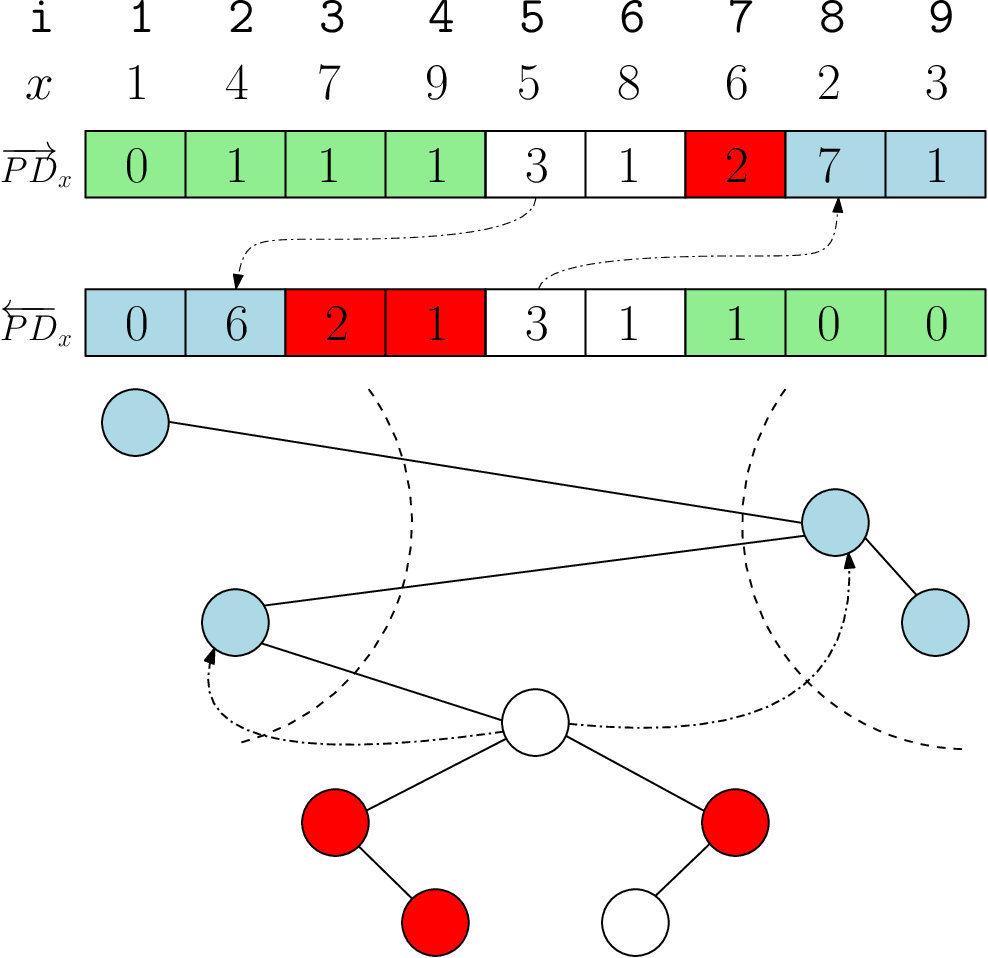}
\end{minipage}
\hfill
\begin{minipage}[t]{0.49\textwidth}
\vspace{0pt}
\includegraphics[scale=0.35]{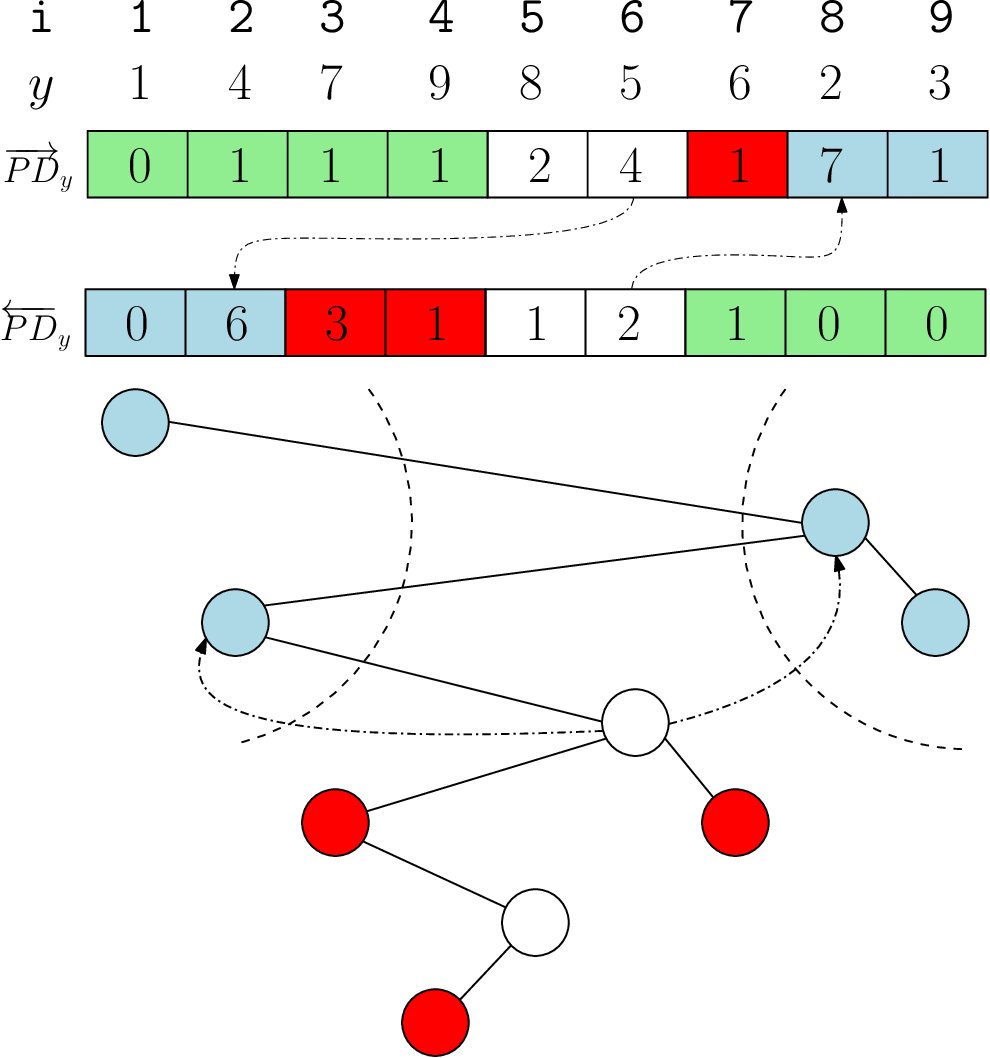}
\end{minipage}
\caption{In this figure, swaps are applied at position \texttt{5} on both $x$ and $y$.
As can be seen on the left part of the Figure, $x[\texttt{5}] < x[\texttt{6}]$, $\ell=\PDD_x[\texttt{5}]$ and $r=\PDG_x[\texttt{5}]$ gives us the position $\texttt{5}+r$ and $\texttt{5}-\ell$ of the first values that are smaller than $x[\texttt{5}]$ and, by extension, smaller than any value between $\texttt{5}-\ell$ and $\texttt{5}+\ell$. Therefore any position smaller than $\texttt{5}-\ell$ in $\PDG_x$ is unaffected by the swap. The same
goes for any position greater than $\texttt{5}+r$ in $\PDD_x$.
On the right part of the figure, we have $y[\texttt{5}] > y[\texttt{6}]$, $\ell=\PDD_y[\texttt{6}]-1$ and $r=\PDG_y[\texttt{6}]+1$.\label{fig:blue_red_appendix}}
\end{figure}

\begin{lemma}\label{lm:propzones}
    Let $x$ and $y$ be two sequences of length $n$ 
    such that $\exists\ i\in \{1, \ldots, n-1\}$ and $y = \tau(x,i)$.
    For all $j\in \{1, \ldots, n\} \setminus \{i,i+1\}$:
    \begin{itemize}
        \item if $\PDD_x[j] \notin \{ j-i-1,j-i\}$, then $\PDD_x[j] = \PDD_y[j]$.
        \item if $\PDG_x[j] \notin \{ i-j, i+1-j\}$, then $\PDG_x[j] = \PDG_y[j]$.
    \end{itemize}        
\end{lemma}
\begin{proof}
    According to Def.~\ref{def:pd},~\ref{def:swap} and~\ref{rev-pd}, we have:
    \begin{itemize}
        \item for all $j<i$, $x[j] = y[j]$: 
        \begin{itemize}
            \item therefore $\PDD_x[j] = \PDD_y[j]$
            \item if $\PDG_x[j] \notin \{ i-j, i+1-j\}$, then by definition
            for all $k \in \{j+1, \ldots, j+\PDG_x[j] -1\}$, both $x[j]$ and $y[j]$ are 
            smaller than $x[k]$ and greater than $x[j+\PDG_x[j]]$ (which is equal to $y[j+\PDG_x[j]]$). Therefore $\PDG_x[j] = \PDG_y[j]$ 
        \end{itemize} 
        \item for all $j>i+1$, $x[j] = y[j]$:        
        \begin{itemize}
            \item therefore $\PDG_x[j] = \PDG_y[j]$
            \item if $\PDD_x[j] \notin \{ j-i-1, j-i\}$, then by definition for all 
            $k \in \{j-\PDD_x[j]+1, \ldots, j-1\}$, both $x[j]$ and $y[j]$ are 
            smaller than $x[k]$ and greater than $x[j-\PDD_x[j]]$ (which is equal to $y[j+\PDD_x[j]]$). Therefore $\PDD_x[j] = \PDD_y[j]$ 
        \end{itemize} 
    \end{itemize}
\end{proof}

\greenzones*
    
\begin{proof}
The proof directly follows from Def.~\ref{def:green} and Lemma~\ref{lm:propzones}.
\end{proof}

\bluezones*


\begin{proof}
Suppose $x[i] < x[i+1]$ (therefore $y[i] > y[i+1]$). 
According to Def.~\ref{def:blue}, the blue zones of $\PDD_x$ and $\PDD_y$ (resp. $\PDG_x$ and $\PDG_y$) are $\PDD_x[i+\bG\ldots n]$ and $\PDD_y[i+\cG +1\ldots n]$ (resp. $\PDG_x[1\ldots i - \aD]$  and $\PDG_y[1\ldots i - (\dD - 1)]$). 
From item~\ref{lm:lesser:item3}  of Lemma~\ref{lm:lesser}, we have $\bG = \cG+1$, meaning that the blue zones of $\PDD_x$ and $\PDD_y$ coincide with each other (resp. item~\ref{lm:lesser:item2} of Lemma~\ref{lm:lesser} for  $\PDG_x$ and $\PDG_y$).

Suppose $\bG > 1$, then there exists a position 
$r=\bG$ such that for all $j \in \{i, \ldots, i+r-1\}$, we have $x[j] > x[i+r]$. 
For each $k \in \{i+r, \ldots, n\}$, either $x[k] \le x[i+r]$, in which case $\PDD_x[k]$ and $\PDD_y[k]$ both point to the green zone and therefore did not change. Otherwise,
$\PDD_x[k]$ and $\PDD_y[k]$ point to at least position $i+r$ and are therefore equal.

If $x[i]<x[i+1]$ then by Def.~\ref{rev-pd} it holds that 
 $x[i+r]<x[i]$ and $x[j]>x[i]$ for all $j\in\{i+1,\ldots,i+r-1\}$.
Let $k\in\{i+r,\ldots,n\}$ be a position the blue zone of $\PDD_x$.
If $x[k]>x[i+r]$ then $\PDD_x[k]\le k-i-r<k-i$ and then
 by Lemma~\ref{lm:propzones} we have $\PDD_y[k]=\PDD_x[k]$.
If $x[k]<x[i+r]$ then $x[k]<x[j]$ for all $j\in\{i,\ldots,i+r\}$ and
 $\PDD_x[k]> k-i$ and then by Lemma~\ref{lm:propzones} we have $\PDD_y[k]=\PDD_x[k]$.

The other cases can be proved in a similar way.

The proof is similar for $x[i] > x[i+1]$.
\end{proof}

\redzones*

\begin{proof}

Let us prove item~\ref{lm:red:item1} when $x[i]<x[i+1]$.
Let $r=\bG$ and
 let $k\in\{i+2,\ldots,i+r-1\}$ be a position in the red zone of $\PDD_x$.

We will consider three cases depending on the value of 
$\PDD_x[k]$:

\begin{enumerate}
    \item If $\PDD_x[k]\notin\{k-i,k-i-1\}$ then by Lemma~\ref{lm:propzones} we have $\PDD_y[k]=\PDD_x[k]$.
    \item If $\PDD_x[k]=k-i$ it means that $x[i]<x[k]$ and
$x[j]>x[k]$
for all $j\in\{i+2,\ldots,k-1\}$.
The same holds for $y$: $y[j]>y[k]$ for all $j\in\{i+2,\ldots,k-1\}$ and $x[i]=y[i+1]<x[k]$
 thus $\PDD_y[k]=\PDD_x[k]-1$.
    \item If $\PDD_x[k]=k-i-1$ it means that $x[i+1]<x[k]$ and $x[j]>x[k]$ for all $j\in\{i+2,\ldots,k-1\}$.
The same holds for $y$: $y[j]>y[k]$ for all $j\in\{i+2,\ldots,k-1\}$ and since $x[i]<x[i+1]<x[k]$
 and $x[i]=y[i+1]$ then $y[i+1] < x[k]$
 thus $\PDD_y[k]=k-i-1=\PDD_x[k]$.
\end{enumerate}

The other items can be proved in a similar way.

\end{proof}

\end{document}